\def\year{2021}\relax
\documentclass[letterpaper]{article} 
\usepackage{aaai21}  
\usepackage{times}  
\usepackage{helvet} 
\usepackage{courier} 
\usepackage[hyphens]{url}
\usepackage{natbib}  
\usepackage[switch]{lineno}  %

\usepackage{amsmath,amsfonts,amsthm}
\usepackage{mathtools}
\usepackage{xcolor}
\usepackage{tikz}

\newtheorem{proposition}{Proposition}

\newtheorem{conjecture}{Conjecture}

\newtheorem{corollary}{Corollary}

\newtheorem{lemma}{Lemma}
\newtheorem{claim}{Claim}
\newtheorem{theorem}{Theorem}
\newtheorem{definition}{Definition}
\theoremstyle{definition}

\newtheorem{innercustomlemma}{Lemma}
\newenvironment{customlemma}[1]
  {\renewcommand\theinnercustomlemma{#1}\innercustomlemma}
  {\endinnercustomlemma}

\newcommand{\opt}{w_{\mathit{opt}}}
\newcommand{\dyp}{{\mathit{dp}}} 
\newcommand{\calP}{{\mathcal{P}}} 
\newcommand{\calF}{{\mathcal{F}}} 
\newcommand{\child}{{\mathit{ch}}} 
\newcommand{\cand}{{\mathit{cand}}}

\title{Proportional Representation under Single-Crossing Preferences Revisited}

\author{
Andrei Constantinescu,  
Edith Elkind\\
}

\affiliations{
University of Oxford\\
eelkind@gmail.com, 
andrei.costin.constantinescu@gmail.com
}

\begin{document}
\maketitle
\pagestyle{plain}

\begin{abstract}
We study the complexity of determining a winning committee under the Chamberlin--Courant voting rule 
when voters' preferences are single-crossing on a line, or, more generally, on a median graph
(this class of graphs includes, e.g., trees and grids). 
For the line, \citeauthor{sc-cc}~(\citeyear{sc-cc}) describe an $O(n^2mk)$ algorithm 
(where $n$, $m$, $k$ are the number of voters, the number of candidates and the committee size, 
respectively); we show that a simple tweak improves the time complexity to $O(nmk)$. We then improve this bound for $k=\Omega(\log n)$ by reducing our problem to the $k$-link path problem for DAGs with concave Monge weights, obtaining a
$nm2^{O\left(\sqrt{\log k\log\log n}\right)}$ algorithm for the general case and a nearly 
linear algorithm for the Borda misrepresentation function. For trees, 
we point out an issue with the algorithm proposed by \citet{cps15},
and develop a $O(nmk)$ algorithm for this case as well. For grids, 
we formulate a conjecture about the structure of optimal solutions, 
and describe a polynomial-time algorithm that finds a winning committee if this conjecture
is true; we also explain how to convert this algorithm into a bicriterial
approximation algorithm whose correctness does not depend on the conjecture.
\end{abstract}

\section{Introduction}
The problem of computing election winners under various voting rules is perhaps the most
fundamental research challenge in computational social choice \cite{comsoc-book}. While this problem
is typically easy for single-winner voting rules (with a few notable exceptions \cite{dodgson,young}), 
for many rules that are supposed to return a {\em set} of winners, the winner determination
problem is computationally demanding. In particular, this is the case for one
of the most prominent and well-studied multiwinner voting rules, namely, 
the Chamberlin--Courant rule \cite{cc83}. Under this rule, each voter is assumed to assign 
a numerical disutility to each candidate; these disutilities are then lifted to sets
of candidates, so that a voter's disutility for a set of candidates $S$ is his minimum disutility
for a member of $S$, and the goal is to identify a committee that minimizes the sum of voters'
disutilities given an upper bound on the committee size (see Section~\ref{sec:prelim} for formal definitions). 
It has been argued that this rule is well-suited for a variety of tasks, ranging from selecting a representative student assembly to deciding which movies to show on a plane
\cite{mw-chapter}. 

Decision problems
related to winner determination under the Chamberlin--Courant rule have been shown 
to be NP-hard even when the disutility function takes a very simple form \cite{prz08,lb11}. 
Accordingly, there is substantial body of work that focuses on identifying special classes
of voters' preferences for which a winning committee can be determined
efficiently. In particular, polynomial-time algorithms have been obtained
for various structured preference domains, such as single-peaked preferences \cite{sp-cc},
single-crossing preferences \cite{sc-cc}, and preferences that are single-peaked on trees, 
as long as the underlying trees have few leaves or few internal vertices \cite{yce13,pe16} 
(see also \cite{trees-full}). These results extend to preferences that are nearly single-peaked or nearly single-crossing, for a suitable choice of distance measure \cite{sp-width,sc-cc,misra}; see also the survey by \citet{structure-chapter} for a summary of results for restricted domains and the survey by \citet{mw-chapter} for a discussion of other approaches to circumventing hardness results for the Chamberlin--Courant rule.

Recently, \citet{kung} and, independently, \citet{cps15} introduced the domain 
of preferences that are singe-crossing on trees, which considerably extends the domain of single-crossing preferences, while sharing some if its desirable properties, such as existence of (weak) Condorcet winners. 
\citet{cps15} also proposed an algorithm for computing the Chamberlin--Courant winners when voters' preferences belong to this domain.
Unfortunately, a close inspection
of this algorithm shows that its running time scales with the number of subtrees
of the underlying tree, which may be exponential in the number of voters; we discuss
this issue in Section~\ref{sec:tree}.

\smallskip

\noindent{\bf Our Contribution}
In this paper, we revisit the problem of computing the winners under the Chamberlin--Courant
rule when the voters' preferences are single-crossing, or, more broadly, single-crossing on a tree.
For the former setting, we observe that a simple tweak of the algorithm of 
\citeauthor{sc-cc}~(\citeyear{sc-cc}) improves the running time from $O(n^2mk)$ to $O(nmk)$.
We then reduce the Chamberlin--Courant winner determination problem 
to the well-studied {\sc DAG $k$-Link Path} problem, and show that the instances of the latter
problem that are produced by our reduction have the {\em concave Monge property}.
We believe that the relationship between the single-crossing property 
and concavity may be of independent interest; further,  
it can be used to show that for $k=\Omega(\log n)$ our problem admits an algorithm that runs in time $nm2^{O\left(\sqrt{\log k\log\log n}\right)}$;
also, for the Borda disutility function (see Section~\ref{sec:prelim}), we
obtain an algorithm that runs in time $O(nm\log(nm))$, i.e., nearly linear in the input size.
This improvement is significant, as in some of the applications we discussed (such as movie 
selection) $k$ can be quite large.

For preferences single-crossing on trees, we design a polynomial-time algorithm that is based on dynamic programming. Interestingly, 
we can achieve a running time of $O(nmk)$ for this case as well.

Finally, we venture beyond trees, and consider preferences single-crossing on grids.
We formulate a conjecture about the structure of optimal solutions in this setting,  
and present a polynomial-time algorithm
whose correctness is guaranteed under this conjecture. We then show how
to transform it into a bicriterial approximation algorithm that is correct
irrespective of the conjecture.

\section{Preliminaries}\label{sec:prelim}
For a positive integer $n$, we write $[n]$ to denote the set $\{1, \dots, n\}$;
given two non-negative integers $n, n'$, we write $[n:n']$ to denote the set
$\{n, \dots, n'\}$.
Given a tree $T$, we write $|T|$ to denote the number of vertices of $T$.

We consider a setting with a set of voters $V$, where $|V|=n$,
and a set of candidates $C=[m]$. Voters rank candidates from best to worst,
so that the preferences of a voter $v$ are given by a linear order $\succ_v$: 
given two distinct candidates $i, j\in C$ we write $i\succ_v j$
when $v$ prefers $i$ to $j$. We write $\calP = (\succ_v)_{v\in V}$;
the list $\calP$ is referred to as a {\em preference profile}.
We assume that we are also given a {\em misrepresentation function}
$\rho: V\times C\to {\mathbb Q}$; we say that $\rho$ is {\em consistent} with $\calP$
if $c\succ_v c'$ implies $\rho(v, c)\le \rho(v, c')$ for each $v\in V$ and all $c, c'\in C$.
Intuitively, the value $\rho(v, c)$ indicates to what extent
candidate $c$ misrepresents voter $v$. 
An example of a misrepresentation function is the {\em Borda misrepresentation function}
$\rho_B$ given by $\rho_B(v, c)=|\{c'\in C: c'\succ_v c\}|$:
this function assigns value $0$ to voter's top choice, 
value $1$ to his second choice, and value $m-1$ to his last choice.

\smallskip

\noindent{\bf Multiwinner Rules\ }
A {\em multiwinner voting rule} maps a profile $\calP$ over a candidate set 
$C$ and a positive integer $k$, $k\le |C|$, to a non-empty collection 
of subsets of $C$ of size at most $k$; the elements of this collection are called 
the {\em winning committees}%
\footnote{Note that we allow committees of size less than $k$, as this simplifies the presentation; all our results hold if each committee is required to be of size exactly $k$}. In this paper, we focus on a family of multiwinner
voting rules known as {\em Chamberlin--Courant} rules \cite{cc83,mw-chapter}. 

An {\em assignment function} is a mapping $w: V\to C$; for each $V'\subseteq V$
we write $w(V')=\{w(v): v\in V'\}$. If $|w(V)|\le k$, 
then $w$ is called a {\em $k$-assignment function}. Given 
a misrepresentation function $\rho$ and
a profile $\calP=(\succ_v)_{v\in V}$, 
the {\em total dissatisfaction of voters in $V$ under a $k$-assignment $w$}
is given by $\Phi_\rho(\calP, w) = \sum_{v\in V}\rho(v, w(v))$. Intuitively, 
$w(v)$ is the representative of voter $v$ in the committee $w(V)$, 
and $\Phi_\rho(\calP, w)$ measures to what extent the voters are dissatisfied
with their representatives. An {\em optimal $k$-assignment function}
for $\rho$ and $\calP$ is a $k$-assignment function that minimizes $\Phi_\rho(\calP, w)$
among all $k$-assignment functions for $\calP$.

The {\em Chamberlin--Courant multiwinner voting rule} takes as input a preference profile $\calP=(\succ_v)_{v\in V}$
over a candidate set $C$, a misrepresentation function $\rho: V\times C\to{\mathbb Q}$ that is consistent with $\calP$ 
and a positive integer $k\le |C|$, and outputs
all sets $W$ such that $W=w(V)$ for some $k$-assignment function $w$ that is optimal 
for $\rho$ and $\calP$. In the {\sc CC-Winner} problem the goal is to find some set $W$
in the output of this rule. 
This problem is known to be NP-complete \cite{prz08}, even if $\rho$ is the Borda misrepresentation
function \cite{lb11}.
We assume that operations on values of $\rho(v, c)$
(such as, e.g., addition) can be performed in unit time; this assumption is 
realistic as the values of $\rho$ are usually small integers.

We say that a $k$-assignment $w$ for a profile $\calP$ and a misrepresentation function $\rho$
is {\em canonical} if $w$ is optimal for $\calP$ and $\rho$,
and for each voter $v\in V$ the candidate $w(v)$ is $v$'s most preferred candidate in $w(V)$. 
If $\rho(v, a)\neq \rho(v, b)$ for all $v\in V$ and all pairs of distinct candidates 
$(a, b)\in C\times C$, then every optimal assignment is canonical; 
however, if it may happen that $\rho(v, a)=\rho(v, b)$ for $a\neq b$, 
this need not be the case. An optimal $k$-assignment $w$ can be transformed into a canonical assignment
$\widehat{w}$ by setting $\widehat{w}(v)$ to be $v$'s most preferred candidate in $w(V)$;
note that this transformation weakly decreases misrepresentation and the committee size.
\smallskip

\noindent{\bf Single-Crossing Preferences\ }
A profile $\calP=(\succ_v)_{v\in V}$ over $C$ is {\em single-crossing (on a line)}
if there is a linear order $\lhd$ on $V$ such that for 
any triple of voters 
$v_1, v_2, v_3$ with $v_1\lhd v_2\lhd v_3$ and every pair of distinct
candidates $(c, c')\in C\times C$ it is not the case that 
$c\succ_{v_1} c'$, $c'\succ_{v_2} c$, and $c\succ_{v_3} c'$. 
Intuitively, if we order the voters in $V$ according to $\lhd$ and go through the list of voters $V$ from left to right, 
every pair of candidates `crosses' at most once.

A profile $\calP=(\succ_v)_{v\in V}$ over $C$ is {\em single-crossing on a tree}
if there exists a tree $T$ with vertex set $V$ that has the following property:
for any triple of voters 
$v_1, v_2, v_3$ such that $v_2$ lies on the path from $v_1$ to $v_3$ in $T$ 
and every pair of distinct candidates $(c, c')\in C\times C$ it is not the case that 
$c\succ_{v_1} c'$, $c'\succ_{v_2} c$, and $c\succ_{v_3} c'$. 
Note that if a profile $\calP$ is single-crossing on a tree $T$ that is a path, then $\calP$ is single-crossing on a line.

We say that an assignment function $w$ for a profile $\calP$ over $C$ that is single-crossing
on a tree $T$ is {\em connected} if for every candidate $c\in C$  
it holds that the inverse image $w^{-1}(c)$ defines a subtree of $T$.
The following lemma shows that, when considering profiles single-crossing on trees, 
we can focus on connected assignment functions.

\begin{lemma}\label{lem:subtree}
For every profile $\calP$ over $C$ that is single-crossing on a tree $T$ and every $k\le |C|$ every canonical $k$-assignment for $\calP$ is connected.
\end{lemma}
\begin{proof}
Let $w$ be a canonical $k$-assignment for $\calP$.
To see that $w$ is connected, fix a candidate $c\in C$ and
let $T'$ be the smallest subtree of $T$ that contains the set 
$w^{-1}(c)$. If $w$ is not connected, then there 
is a voter $v$ in $T'$ such that $w(v) = c'$, $c'\neq c$, and
deleting $v$ would disconnect $T'$. Then 
there are two voters $x, y \in T' \cap w^{-1}(c)$ for which the unique simple $x$--$y$ 
path contains $v$. Since $w$ is a canonical assignment,  
we have $c \succ_x c', c \succ_y c'$, but $c' \succ_{v} c$, 
a contradiction with $\calP$ being single-crossing on~$T$.
\end{proof}

The next lemma establishes a monotonicity property of canonical assignments
that will be useful for our analysis.

\begin{lemma} \label{lem:increasing}
Consider a profile $\calP=(\succ_v)_{v\in V}$ that is single-crossing on a tree $T$, 
and suppose that voter $v$ ranks the candidates as $1\succ_v\ldots\succ_v m$.
Then, every canonical $k$-assignment for $\calP$
is non-decreasing along every simple path in $T$ starting at $v$.
\end{lemma}
\begin{proof}
Consider a canonical $k$-assignment $w$. Let $P$ be a simple path starting at voter~$v$ and let $x$, $y$ be two voters on $P$ such that $x$ precedes $y$ on $P$. Suppose $w(x) = a$, $w(y) = b$ with $a > b$. Then
$b \succ_v a$, $a \succ_x b$ and $b \succ_y a$, 
a contradiction with $\calP$ being single-crossing on $T$.
\end{proof}

The concept of single-crossing preferences can be extended beyond lines and trees
to a class of graphs known as {\em median graphs}, defined below \cite{median}.

\begin{definition} 
An undirected graph $G$ is called a \emph{median graph} if for every triple of vertices 
$a, b, c$ there exists a unique vertex $m(a, b, c)$, called the \emph{median} 
of $a, b, c$, which 
is simultaneously on one $a$--$b$, one $b$--$c$ and one $c$--$a$ shortest path.
Given a median graph $G$ with vertex set $V$, we say that a preference profile 
$\calP=(\succ_v)_{v\in V}$ is \emph{single-crossing with respect to $G$} if for every pair of 
voters $s, t$ and for every shortest $s$--$t$ path $X$ in $G$, the restriction of $\calP$ to 
the voters that appear on $X$ is single-crossing on the line with respect to the natural 
order induced by $X$.
\end{definition}
It is not hard to check that paths and trees are median graphs. Another useful class of 
median graphs are {\em grid graphs}: a {\em two-dimensional grid graph} is a graph with 
vertex set $[n_1]\times [n_2]$ for some positive integers $n_1, n_2$ such that there is an 
edge between two vertices $(i, j)$ and $(i', j')$ if and only if $|i-i'|+|j-j'|=1$. (This 
definition and our results for grids extend beyond two dimensions, 
but for readability we focus on the 
two-dimensional case).

We will be interested in solving {\sc CC-Winner} if voters' preferences
are single-crossing on a line, on a tree or on a grid; we denote these special
cases of our problem by {\sc CC-Winner-SC}, {\sc CC-Winner-SCT}, and {\sc CC-Winner-SCG},
respectively. We assume that the respective ordering/tree/grid is given as part of the input;
this assumption is without loss of generality as the respective graph
can be computed from the input profile in polynomial time \cite{df,kung,cps15,median}. 
\smallskip

\noindent{\bf Rooted Trees and DAGs\ }
A {\em rooted tree} is a finite tree with a designated root vertex $r$. 
We say that a vertex $u$ is a {\em parent} of $v$ (and $v$ is a {\em child} of $u$)
if $u$ and $v$ are adjacent and $u$ lies on the path from $v$ to $r$. A vertex
with no children is called a {\em leaf}. We denote the number of children of vertex 
$v$ by $n_v$, and represent the set of children of $v$ as an array
$\child[v, 1], \dots, \child[v, n_v]$.
We write $T_v$ to denote the subtree of $T$ with vertex set 
$\{u: \text{ the path from $u$ to $r$ goes through $v$}\}$.
Similarly, for each $v\in V$ and $i\in [1:n_v+1]$, 
let $T_{v, i}$ be the subtree obtained by starting with $T_v$ and removing
the subtrees $T_{\child[v, 1]}, \dots, T_{\child[v, i-1]}$.
Observe that $T_{v,1} = T_v$ and that $T_{v,n_v+1}$ 
is just the singleton vertex $v$.

A \emph{directed acyclic graph} (DAG) is a finite oriented graph whose vertices can be totally 
ordered so that the tail of each arc precedes its head 
in the ordering.  
All DAGs we will consider have the set $[0:n]$ as their set of vertices, and are 
ordered with respect to the natural ordering $<$. A DAG is said to be {\em weighted} 
if its arcs are given real values by a function $\omega$.
A weighted DAG on vertex set $[0:n]$ satisfies the \emph{concave Monge} property if 
for all vertices $i$, $j$ such that $0 < i + 1 < j < n$ it holds that 
$\omega(i, j) + \omega(i + 1, j + 1) \leq \omega(i, j + 1) + \omega(i + 1, j)$.
We refer to the weight function $\omega$ itself as being {\em concave Monge} 
in this case.

\section{Improved Algorithms for Single-Crossing Preferences}\label{sec:line}
We start by considering the setting where the voters' preferences are single-crossing
on a line. We assume without loss of generality that the voter order $\lhd$
is given by $v_1\lhd\ldots\lhd v_n$ and that the first voter ranks the candidates 
in $C=[m]$ as $1\succ_{v_1}\ldots\succ_{v_1} m$. 

The following lemma is implicit in the work of 
\citeauthor{sc-cc}~(\citeyear{sc-cc}), 
and can be seen as an instantiation of Lemmas~\ref{lem:subtree} and~\ref{lem:increasing}.

\begin{lemma}\label{lem:path_mon} 
For every canonical assignment $\opt$ for {\sc CC-Winner-SC}
and every pair of voters $v_i, v_j$ with $i<j$ 
it holds that $\opt(v_i) \leq \opt(v_j)$.
\end{lemma}

\citeauthor{sc-cc}~(\citeyear{sc-cc}) use this lemma to develop a dynamic programming algorithm
for {\sc CC-Winner-SC} that runs in time $O(n^2mk)$.
We will now present a faster dynamic programming algorithm that uses auxiliary variables.

\begin{theorem}\label{thm:line-nmk}
Given an instance of {\sc CC-Winner-SC} with $n$ voters, $m$ candidates and committee size $k$, 
we can compute an optimal solution in time $O(nmk)$. 
\end{theorem}
\begin{proof}
We will explain how to compute the minimum dissatisfaction; a winning committee
can then be computed using standard dynamic programming techniques.

We define the following subproblems for each $i\in [n]$, $c\in [m]$ and each 
$\ell=1, \dots, \min\{k, m-c+1, n-i+1\}$:
\begin{itemize}
\item 	
let $\dyp_0[i, \ell, c]$ be the minimum dissatisfaction of voters in $\{v_i, \dots, v_n\}$ 
                          for a size-$\ell$ committee that is contained in $[c : m]$; 
\item 	
let $\dyp_1[i, \ell, c]$ be the minimum dissatisfaction of voters in $\{v_i, \dots, v_n\}$ 
                          for a size-$\ell$ committee that is contained in $[c : m]$ and represents $v_i$ by $c$. 
\end{itemize}

To simplify presentation, we assume $\dyp_f[-, -, -] = \infty$ for $f \in\{0, 1\}$ 
and $i, c, \ell$ not satisfying the conditions.

We have  $\dyp_1[n, 1, c] = \rho(v_n, c)$ for each $c\in C$.
Also, $\dyp_0[n, 1, m] = \rho(v_n, m)$, and for $c<m$ we have
         $\dyp_0[n, 1, c] = \min\{\dyp_1[n, 1, c], \dyp_0[n, 1, c + 1]\}$.
         
\noindent For $i=n-1, \dots, 1$ we have the following recurrence:
\begin{align*}
    \dyp_1[i, \ell, c] &=
          \rho(v_i, c) \\&+
         \min\{ \dyp_1[i + 1, \ell, c],
         \dyp_0[i + 1, \ell - 1, c + 1]);\\
    \dyp_0[i, \ell, c] &= 
    \min\{\dyp_1[i, \ell, c], \dyp_0[i, \ell, c + 1]\}.
\end{align*}
This recurrence enables us to compute all values $\dyp_f[-, -, -]$ for $f\in\{0, 1\}$; 
the minimum dissatisfaction in our instance is given by $\min_{1\le \ell\le k}\dyp_0[1, \ell, 1]$.
The dynamic program has $O(nmk)$ entries; each entry can be computed in time
$O(1)$ given the already-computed entries.
\end{proof}
To improve over the $O(nmk)$ bound, we will reduce {\sc CC-Winner-SC}
to the well-studied {\sc DAG $k$-link Path} problem with Monge concave weights 
(see, e.g., \citet{bein_larmore_park}), and use 
the powerful machinery developed for it to obtain faster algorithms for 
our setting.

\begin{definition}\label{def:k-path} 
Given a DAG with an arc weight function $\omega$ and two designated vertices $s$ and $t$, the 
{\sc $k$-Link Path} problem ($k$-LPP) asks for a minimum total weight path starting at $s$, ending at $t$ 
and consisting of exactly $k$ arcs.
\end{definition}
There are a number of algorithmic results for the {\sc $k$-Link Path} problem assuming a 
concave Monge weight function
\cite{bein_larmore_park,worse_baruch,baruch_klink_monge}. We will first present our reduction, and 
then discuss the implications for {\sc CC-Winner-SC}.

Given an instance of {\sc CC-Winner-SC} with a preference profile $\calP=(\succ_v)_{v\in V}$ over 
$C=[m]$, we construct a DAG $G$ with vertex set $[0:n]$ and the weight function $\omega$ given by
\begin{equation} \label{eq:omega}
    \omega(i, j) = \min_{c \in C}{\left(\rho(v_{i + 1}, c) + \ldots + \rho(v_j, c)\right)}.
\end{equation}
That is, $\omega(i, j)$ represents the minimum total dissatisfaction that voters in $\{v_{i+1}, \dots, v_j\}$ 
derive from being represented by a single candidate $c$. Let 
$\cand(i, j)$ be some candidate in $\arg\min_{c\in C}\left(\rho(v_{i + 1}, c) + \ldots + \rho(v_j, c)\right)$. 

First, we observe
that an optimal solution to {\sc CC-Win\-ner-SC} corresponds to a minimum cost path in $k$-LPP.

\begin{theorem}\label{thm:link}
The minimum cost of a $k$-link $0$--$n$ path in $G$ with respect to $\omega$ is equal to the minimum total dissatisfaction for $\calP$ and $k$.
\end{theorem}
\begin{proof}
Let $P = 0 \rightarrow \ell_1 \rightarrow \ldots \rightarrow \ell_{k - 1} \rightarrow n$ 
be a minimum cost $k$-link path in $G$. 
Then $P$ induces an assignment of candidates to voters: 
if $P$ contains an arc $(i, j)$ we assign candidate $\cand(i, j)$
to voters $v_{i+1}, \dots, v_j$.
The total dissatisfaction under this assignment equals to the weight of~$P$. 

Conversely, let $\opt$ be a canonical $k$-assignment. By Lemma~\ref{lem:path_mon} we know that $\opt$ partitions the voters into contiguous subsequences. To build a path $P$ in $G$, we proceed as follows.
For every maximal contiguous subsequence of voters $v_{i+1}, \dots, v_j$ represented by 
the same candidate in $\opt$, we add the arc $i \rightarrow j$ to $P$. 
By construction, the resulting set of arcs forms a $k$-link path from $0$ to $n$, 
and its total weight is at most $\Phi_\rho(\calP, \opt)$.
\end{proof}

Note, however, that Theorem~\ref{thm:link} is insufficient for our purposes:
the efficient algorithms for $k$-LPP require the weight function $\omega$
to have the concave Monge property, so we need to prove that the reduction in
Theorem~\ref{thm:link} always produces such instances of $k$-LPP.

We say that an instance of {\sc CC-Winner-SC} is {\em concave Monge} 
if the reduction in Theorem~\ref{thm:link} maps it to an instance of $k$-LPP with 
the concave Monge property. Thus, we need to prove that
each instance of {\sc CC-Winner-SC} is concave Monge. 
To this end, we will first argue that if there is an instance of {\sc 
CC-Winner-SC} that is not concave Monge, 
then there exists such an instance with three voters.
Then we prove that every three-voter instance is concave Monge.

\begin{lemma}\label{lem:3voters}
If there exists a non-concave Monge instance of {\sc CC-Winner-SC}, 
then there exists a non-concave Monge instance of {\sc CC-Winner-SC} 
with three voters.
\end{lemma}
\begin{proof}
Consider a non-concave Monge instance of {\sc CC-Winner-SC} with $n\neq 3$ voters.
Note that $n \geq 4$: indeed, for $n<3$ there is no pair of vertices $i, j$ 
that satisfies $0<i+1<j<n$.
We can assume that the $(i, j)$ pair that violates the concave Monge property is $(0, n - 1)$: 
otherwise we could just remove all voters before $v_{i+1}$ and all voters after $v_j$. 
Thus, we have $\omega(0, n-1) + \omega(1, n) > \omega(0, n) + \omega(1, n-1)$.
Recall that
\begin{align}
\omega(0, n-1) &= \min_{c \in C}{(\rho(v_1, c) +\ldots +  \rho(v_{n-1}, c))};\label{eq:o1}\\
\omega(1, n)   &= \min_{c \in C}{(\rho(v_2, c) + \ldots + \rho(v_n, c))};\label{eq:o2}\\
\omega(0, n)   &= \min_{c \in C}{(\rho(v_1, c) + \ldots +\rho(v_n, c))};\label{eq:o3}\\
\omega(1, n-1)  &= \min_{c \in C}{(\rho(v_2, c) +\ldots + \rho(v_{n-1}, c))}.\label{eq:o4}
\end{align}

Now, consider a three-voter profile $(\succ_x, \succ_y, \succ_z)$ with misrepresentation
function $\rho'$ constructed as follows. Set $\succ_x\,=\,\succ_{v_1}$, $\succ_z\,=\,\succ_{v_n}$
and $\rho'(x, c)=\rho(v_1, c)$, $\rho'(z, c) = \rho(v_n, c)$ for all $c\in C$.
Further, set $\rho'(y, c)=\rho(v_2, c) + \rho(v_3, c) + \ldots + \rho(v_{n - 1}, c)$. 
and let $a\succ_y b$ if and only if 
$\rho'(y, a) < \rho'(y, b)$ or
$\rho'(y, a) = \rho'(y, b)$ and $a\succ_{v_1} b$.
One can verify that $\succ_y$ is a linear order.
Moreover, we claim that the profile $(\succ_x, \succ_y, \succ_z)$ 
is single-crossing with respect to the voter order $x\lhd y\lhd z$.
Indeed, consider two distinct candidates $a, b$. If $x$ and $z$ disagree
on $(a, b)$, then in 
$(\succ_x, \succ_y, \succ_z)$ 
candidates $a$ and $b$ cross
at most once, irrespective of $y$'s preferences. Now suppose that 
$x$ and $z$ agree on $(a, b)$; for concreteness, suppose that $a\succ_x b$, 
$a\succ_z b$. As the input profile is single-crossing, all other voters 
also prefer $a$ to $b$ and hence 
$\rho(v_2, a) + \rho(v_3, a) + \ldots + \rho(v_{n - 1}, a) \le 
\rho(v_2, b) + \rho(v_3, b) + \ldots + \rho(v_{n - 1}, b)$, 
in which case $a\succ_y b$. Hence, $(\succ_x, \succ_y, \succ_z)$ 
is indeed single-crossing.

Now, we can rewrite equations~\eqref{eq:o1}--\eqref{eq:o4} as
\begin{align*}
\omega(0, n-1) &= \min_{c\in C}\left({\rho'(x, c)+\rho'(y, c)}\right);\\
\omega(1, n)   &= \min_{c\in C}\left({\rho'(y, c)  +\rho'(z, c)}\right);\\
\omega(0, n)   &= \min_{c\in C}\left(\rho'(x, c)+\rho'(y, c)+\rho'(z, c)\right);\\
\omega(1, n-1) &= \min_{c\in C}{\rho'(y, c)}.
\end{align*}
This shows that the instance of {\sc CC-Winner-SC} formed by $x$, $y$ and $z$ together with $\rho'$ is also non-concave Monge.
\end{proof}

\begin{proposition}
Every instance of {\sc CC-Winner-SC} is concave Monge.
\end{proposition}
\begin{proof}
By Lemma~\ref{lem:3voters}, it suffices to consider instances with three voters.
Thus, consider a three-voter profile that is single-crossing with respect to the voter
order $v_1\lhd v_2 \lhd v_3$. We need to argue that 
$$
\omega(0, 2)+\omega(1, 3)\le \omega(0, 3)+\omega(1, 2).
$$
Let $a=\cand(1, 2)$; we can assume that $a$ is the top candidate of the second voter. Also, let $b=\cand(0, 3)$, $c_1=\cand(0, 2)$, $c_2=\cand(1, 3)$.

Suppose first that $b=a$. Then 
$$
\omega(0, 3)+\omega(1, 2)=\rho(v_1, a)+2\rho(v_2, a)+\rho(v_3, a).
$$ 
On the other hand, 
$c_1=\cand(0, 2)$ implies that 
$$\omega(0, 2) = \rho(v_1, c_1)+\rho(v_2, c_1)\le \rho(v_1, a)+\rho(v_2, a),
$$ and 
$c_2=\cand(1, 3)$ implies that
$$
\omega(1, 3) = \rho(v_2, c_2)+\rho(v_3, c_2)\le \rho(v_2, a)+\rho(v_3, a).
$$
Adding up these inequalities, we obtain the desired result.

Now, suppose that $b\neq a$. 
Then 
$$\omega(0, 3)+\omega(1, 2)=\rho(v_1, b)+\rho(v_2, b)+\rho(v_2, a)+\rho(v_3, b).
$$
As the second voter ranks $a$ first, she ranks $b$ below $a$. Hence, by the single-crossing property, at least one other voter prefers $a$ to $b$; we can assume without loss of generality that $a\succ_{v_3} b$.
Thus, $\rho(v_3, b)\ge \rho(v_3, a)$ and hence 
$$
\omega(0, 3)+\omega(1, 2)\ge \rho(v_1, b)+\rho(v_2, b)+\rho(v_2, a)+\rho(v_3, a).
$$
Now, $c_1=\cand(0, 2)$ implies that 
$$
\omega(0, 2) = \rho(v_1, c_1)+\rho(v_2, c_1)\le \rho(v_1, b)+\rho(v_2, b),
$$ 
and 
$c_2=\cand(1, 3)$ implies that
$$
\omega(1, 3) = \rho(v_2, c_2)+\rho(v_3, c_2)\le \rho(v_2, a)+\rho(v_3, a).
$$ 
Again, adding up these inequalities, we obtain the desired result. 
\end{proof}

It now follows that any fast algorithm for $k$-LPP with the concave Monge property 
translates into an algorithm for {\sc CC-Winner-SC}, slowed down by a factor of $O(m)$ 
required for computing arc weights%
\footnote{Computing all arc weights in advance would be 
too expensive. Instead, we precompute the values $\sum_{\ell=1}^j \rho(v_\ell, c)$ for all 
$j\in [n]$, $c\in C$ in time $O(nm)$; then, when the algorithm needs to know
$\omega(i, j)$, we compute $\sum_{\ell=i+1}^j \rho(v_\ell, c)$ for each $c$
in time $O(1)$ as a difference of two precomputed quantities, and then compute the minimum over 
$C$ in time $O(m)$.}.
 
Now, if individual dissatisfactions are non-negative integers in range $[0:U]$ (e.g.~$U 
= m$ for the Borda misrepresentation function), then the weakly-polynomial algorithm of
\citet{bein_larmore_park} and 
\citet{worse_baruch} leads to an $O(nm \log(nU))$ algorithm for 
{\sc CC-Winner-SC}. Alternatively, we can use the strongly-polynomial time algorithm 
of \citet{baruch_klink_monge} to get a runtime of 
$nm2^{O\left(\sqrt{\log k\log\log n}\right)}$, which improves on our earlier bound of $O(nmk)$ for 
$k = \Omega(\log n)$.
We summarize these results in the following theorem.
\begin{theorem}
Given an instance of {\sc CC-Winner-SC} with $n$ voters, $m$ candidates 
and committee size $k$, where $k=\Omega(\log n)$,
we can compute an optimal solution in time $nm2^{O\left(\sqrt{\log k\log\log n}\right)}$.
Moreover, if $\rho$ is the Borda misrepresentation function, 
we can compute an optimal solution in time $O(nm\log(nm))$.
\end{theorem}

\section{Preferences Single-Crossing on a Tree}\label{sec:tree}
\citeauthor{cps15}~[\citeyear{cps15}] present an algorithm for {\sc CC-Winner-SCT}
that proceeds by dynamic programming, 
building a solution for the entire tree from solutions for various subtrees. Entries of their dynamic program are indexed by subtrees of the input tree, and on some instances the algorithm may need to consider all subtrees containing the root; a tree on $n$ vertices can have $2^{\Omega(n)}$
such subtrees. We present a detailed example in the appendix.%
\footnote{We have contacted the authors of the paper, and they have acknowledged this issue.}

\subsection{A Dynamic Programming Solution}
We will now present a different dynamic programming algorithm, which provably 
runs in polynomial time.
This algorithm, too, builds a solution iteratively by considering subtrees of the original
tree, but it proceeds in such a way that it only needs to consider polynomially many subtrees.

Fix a misrepresentation function $\rho$, and consider a profile $\calP=(\succ_v)_{v\in V}$, $|V|=n$, 
over $C=[m]$ that is single-crossing on a tree $T$.
We will view $T$ as a rooted tree with $v_1$ as its root, and
assume without loss of generality that $1\succ_{v_1}\ldots\succ_{v_1} m$.

We first reformulate our problem as a tree partition problem using Lemma~\ref{lem:subtree}.

\begin{definition}\label{def:treepartition}
A {\em $p$-partition} of $T$ is a partition of $T$ 
into $p$ subtrees $\calF = \{F_1, \dots, F_p\}$.
A $p$-assignment $w:V\to C$ for a profile $\calP=(\succ_v)_{v\in V}$
that is single-crossing on a tree $T$ is a {\em $p$-tree assignment}
if there is a $p$-partition $\{F_1, \dots, F_p\}$ of $T$
such that $w(v)=w(v')$ for each $\ell\in [p]$ and $v, v'\in F_\ell$. 
In the {\sc Chamberlin--Courant Tree Partition} (CCTP) problem 
the goal is to find a value $p\in [k]$ and a $p$-tree assignment $\opt$, together with the associated $p$-partition $\calF_{\mathit{opt}}$, such that
$\opt$ minimizes $\Phi_\rho(\calP, w)$ over all $p\in [k]$ and all 
$p$-tree assignments for $\calP$.
\end{definition}

\noindent By Lemma~\ref{lem:subtree}, an optimal assignment
for CCTP is an optimal assignment for the associated 
instance of {\sc CC-Winner-SCT}.

We start by presenting a dynamic programming algorithm for CCTP and proving a bound
of $O(nmk^2)$ on its running time; later, we will improve this bound 
to $O(nmk)$.

\begin{theorem}\label{thm:cctp-quad}
Given an instance of CCTP with $n$ voters, $m$ candidates and committee size $k$,
we can compute an optimal solution in time $O(nmk^2)$. 
\end{theorem}
\begin{proof}
We will explain how to find the value of an optimal solution in time $O(nmk^2)$;
an optimal solution can then be recovered using standard dynamic programming techniques.

We define the following subproblems for each $v\in V$ and $c\in [m]$.

\begin{itemize}
\item 	For each $\ell=1, \dots, \min\{k, |T_v|\}$,\\ let
	$\dyp_0[v, \ell, c]$  be the minimal dissatisfaction of voters in $T_v$ 
                             that can be achieved by partitioning $T_{v}$ into 
                             $\ell$ subtrees using only candidates in $[c : m]$ 
                             as representatives. 
\item 	For each $\ell=1, \dots, \min\{k, |T_v|\}$, \\ let
	$\dyp_1[v, \ell, c]$  be the minimal dissatisfaction of voters in $T_v$ 
			     that can be achieved by partitioning $T_{v}$ into 
			     $\ell$ subtrees using only candidates in $[c : m]$ 
		             as representatives, with voter $v$ represented by candidate $c$.
\item 	For each $i\in[n_v+1]$, and each $\ell=1, \dots, \min\{k, |T_{v, i}|\}$, \\ let
	$\dyp_2[v, i, \ell, c]$ be the minimal dissatisfaction of voters in $T_{v, i}$ 
			     that can be achieved by partitioning $T_{v, i}$ into 
			     $\ell$ subtrees using only candidates in $[c : m]$ 
                             as representatives, with voter $v$ represented by candidate $c$.
\end{itemize}

To simplify presentation, we assume these quantities to take value $\infty$ for $v, i, \ell, c$ not satisfying the conditions.

Clearly, the value of an optimal solution to our instance of CCTP is 
$\min_{\ell\in[k]}\dyp_0[v_1, \ell, 1]$.
It remains to explain how to compute the intermediate quantities.

The following observations are immediate from the definitions of $\dyp_0$, $\dyp_1$, $\dyp_2$:
\begin{align}
    \label{l1}& \dyp_2[v, n_v + 1, 1, c] = \rho(v, c), \\
    \label{l2}& \dyp_1[v, \ell, c] = \dyp_2[v, 1, \ell, c], \\
    \label{l3}& \dyp_0[v, \ell, c] = \min\{\dyp_1[v, \ell, c], \dyp_0[v, \ell, c + 1]\}.
\end{align}
The next lemma explains how to compute $\dyp_2$.

\begin{lemma}\label{lem:recurrence} 
Let $u$ be the $i$-th child of $v$, and let $s= |T_{v, i+1}|$.
Then $\dyp_2[v, i, \ell, c] = \min\{\textit{DIFF}, \textit{SAME}\}$, where

\begin{align*}
    \textit{DIFF} &= \min\{\dyp_0[u, t, c + 1] + \dyp_2[v, i + 1, \ell - t, c]: \\
                           &1\leq t      \leq \min\{\ell, |T_u|\},\
                           1\leq \ell-t  \leq \min\{\ell, s\} \}, \\
    \textit{SAME} &= \min\{\dyp_1[u, t, c]     + \dyp_2[v, i + 1, \ell - t + 1, c]: \\
                           &1\leq t      \leq \min\{\ell, |T_u|\},\ 
                           1\leq \ell-t+1 \leq \min\{\ell, s\} \}.
\end{align*}
\end{lemma}
\begin{proof}
Throughout this proof we will make implicit use of Lemma~\ref{lem:increasing} 
to justify always having candidates with a higher index be further down in the tree.
Let $(\calF_{\mathit{opt}}, \opt)$ be an optimal connected $\ell$-tree partition of $T_{v, i}$ 
such that voter $v$ is assigned candidate $c$. 
There are two cases to consider:

\begin{itemize}
\item
	Voter $u$ is represented by a candidate $c' > c$. 
	Then each subtree in $\calF_{\mathit{opt}}$ 
	is either fully contained in $T_u$ or fully contained in $T_{v, i + 1}$, 
	so there is a number $t\in [\ell]$ such that 
	$\calF_{\mathit{opt}}$ partitions $T_u$ into $t$ subtrees and 
	$T_{v, i + 1}$ into $\ell - t$ subtrees. Hence, to minimize 
	dissatisfaction, we take the minimum 
	over all $t\in [\ell]$. For a fixed $t\in[\ell]$ we choose 
	(i) an optimal $t$-partition of $T_u$ that uses candidates 
	in $[c + 1 : m]$ only and (ii) an optimal $(\ell-t)$-partition of $T_{v, i + 1}$ 
	that uses candidates in $[c : m]$ only and assigns $c$ to $v$. 
	The optimal values of the former 
	and the latter are given by $\dyp_0[u, t, c + 1]$ 
	and $\dyp_2[v, i + 1, \ell - t, c]$, respectively.
\item
	Voter $u$ is represented by candidate $c$. In this case the 
	subtree in $\calF_{\mathit{opt}}$ that contains $v$ may partly reside in both $T_u$ and 
	$T_{v, i + 1}$. Therefore, there is a number $t\in [\ell]$ 
        such that $\calF_{\mathit{opt}}$ 
	partitions $T_u$ into $t$ subtrees and $T_{v, i + 1}$ into $\ell - t+ 1$ subtrees. 
	Once again, we take the minimum over all $t\in [\ell]$. For a fixed $t\in [\ell]$ 
	we choose (i) an optimal $t$-partition of $T_u$ that uses candidates in $[c:m]$ only 
	and assigns $c$ to $u$ and (2) an optimal $(\ell-t+1)$-partition of $T_{v, i + 1}$ 
	that uses candidates in $[c : m]$ only and assigns $c$ to $v$. 
	By definition, the optimal values 
	for these subproblems are given by $\dyp_1[u, t, c]$ and 
	$\dyp_2[v, i + 1, \ell - t + 1, c]$, respectively.
\end{itemize}
\end{proof}
Note that the assignment implicitly computed by our dynamic programming algorithm is not 
necessarily connected; however, this is not required for optimality.

Our dynamic program proceeds from the leaves to the root of $T$, 
computing the quantities $\dyp_0, \dyp_1$ and $\dyp_2$; 
we process a vertex after its children have been processed.
Computing all these quantities is trivial if $v$ is a leaf;
if $v$ is not a leaf, we first compute $\dyp_2[v, i, \ell, c]$ for all $i=|n_v|+1, \dots, 1$
and all relevant values of $\ell$ and $c$ using~\eqref{l1} and Lemma~\ref{lem:recurrence}, 
and then compute $\dyp_1[v, \ell, c]$ (using~\eqref{l2}) and $\dyp_0[v, \ell, c]$
(using~\eqref{l3}) for $c=m, \dots, 1$.
To bound the running time, note that
\begin{itemize}
    \item there are $O(nmk)$ subproblems of the form $\dyp_0[-, -, -]$ and $\dyp_1[-, -, -]$, 
	  each requiring constant time to solve;
    \item there are $O(nmk)$ subproblems of the form $\dyp_2[-, -, -, -]$. 
          This is because pairs of the form ($v$, $i$) such that $1 \leq i \leq n_v$ 
          correspond to edges of the tree and there are precisely $n-1$ of them. 
          Each of these subproblems can be solved in time $O(k)$ by Lemma~\ref{lem:recurrence}.
    \item The tree can be traversed in time $O(n)$.
\end{itemize}
Altogether, we get a time bound of $O(nmk^2)$.
\end{proof}

\subsection{Tighter Analysis of the Running Time}
We will now show how to improve the bound on the running time 
of our algorithm to $O(nmk)$. 
To do so, it suffices to establish that
all subproblems of the form $\dyp_2[-, -, -, -]$ can be solved in time $O(nmk)$. 

The following technical lemma is an important building block in our analysis (see the online appendix for the proof).

\begin{lemma}\label{lem:product}
Consider a voter $v$ and a candidate $c$, 
and let $u$ be the $i$-th child of $v$ for some $i\in [n_v]$.
Then all subproblems of the form 
$\dyp_2[v, i, -, c]$ can be solved in time 
$O(\min\{k, |T_u|) \cdot \min\{k, |T_{v, i + 1}|\})$. 
\end{lemma}
\begin{proof} 
We first look at the time required to compute $\textit{DIFF}$ 
from Lemma~\ref{lem:recurrence} for all $\ell$ with 
$1 \leq \ell \leq \min\{k, |T_{v, i}|\}$. For each such $\ell$ 
define a set of pairs
\begin{align*}
    M_{\ell} = \{(\ell, t) :\ &1 \leq t \leq \min\{\ell, |T_u|\}\text{ and}\\ 
                              &1 \leq \ell - t \leq \min\{\ell, |T_{v, i + 1}|\}\};
\end{align*}
the pairs in $M_\ell$ appear in the expression for $\textit{DIFF}$ 
when solving subproblems of the form 
$\dyp_2[v, i, \ell, c]$. Clearly, the time it takes to compute $\textit{DIFF}$ for all  
values of $\ell$ is asymptotically bounded by the total size of these sets. 
Now, apply the bijective map 
$(\ell, t) \mapsto (\ell - t, t)$ to each $M_{\ell}$ to get a new collection of sets
\begin{align*}
    M_{\ell}' = \{(\ell - t, t) :\ &1 \leq t \leq \min\{\ell, |T_u|\}\text{ and}\\ 
                                   &1 \leq \ell - t \leq \min\{\ell, |T_{v, i + 1}|\}\}
\end{align*}
\noindent with the same total cardinality. Their union is a set
\begin{alignat*}{3}
M' &= \{(\ell - t, t) :\ &&1 \leq \ell \leq \min\{k, |T_{v, i}|\},\\
   &                     &&1 \leq t \leq \min\{\ell, |T_u|\}\text{ and}\\ 
   &                     &&1 \leq \ell - t \leq \min\{\ell, |T_{v, i + 1}|\}
\}
\end{alignat*}
\begin{alignat*}{3}
   &= \{(\ell - t, t) :\ &&1 \leq \ell - t + t \leq \min\{k, |T_{v, i}|\},\\
   &                     &&1 \leq t \leq \min\{\ell - t + t, |T_u|\}\text{ and}\\
   &                     &&1 \leq \ell - t \leq \min\{\ell - t + t, |T_{v, i + 1}|\}\}
\end{alignat*}
\begin{alignat*}{3}
   &= \{(x, y) :\                &&1 \leq x + y \leq \min\{k, |T_{v, i}|\},\\
   &                             &&1 \leq y \leq \min\{x + y, |T_u|\}\text{ and}\\ 
   &                             &&1 \leq x \leq \min\{x + y, |T_{v, i + 1}|\}\}
\end{alignat*}
\begin{alignat*}{3}
   &= \{(x, y) :\                &&x + y \leq k, 1 \leq y \leq |T_u|, 
                                   1 \leq x \leq |T_{v, i + 1}|\}.
\end{alignat*}
\noindent To obtain the last equality we make use of the fact that 
$|T_{v, i}| = |T_u| + |T_{v, i + 1}|$. We can now relax the constraints 
$x, y\ge 1$, $x + y \leq k$ 
to $1 \leq x, y \leq k$ to conclude that 
$|M'| \leq \min\{k, |T_u|\} \cdot \min\{k, |T_{v, i + 1}|\}$. 
Similar analysis shows that the same bound also holds for $\textit{SAME}$, hence completing the proof.
\end{proof}

Before proving the stronger $O(nmk)$ bound, we first show an easier bound of $O(n^2m)$, which is 
tight when $k = n$ and better than $O(nmk^2)$ whenever $k = \omega(n^{1/2})$. Proving this is both 
informative in itself, helping to build intuition, and will also provide us with a tool useful for 
the general argument. The $O(n^2m)$ bound is immediate from the following lemma (inspired by \citet{barricades}).

\begin{lemma}\label{lem:n_squared} 
For each candidate $c \in C$ solving all subproblems of the form 
$\dyp_2[-, -, -, c]$ using Lemma~\ref{lem:recurrence} takes time $O(n^2)$.
\end{lemma}
\begin{proof} 
By Lemma~\ref{lem:product}, the time required to solve all subproblems of the form 
$\dyp_2[-, -, -, c]$ is asymptotically bounded by
\[
\sum_{v\in V, 1 \leq i \leq n_v}
\left(|T_{\child[v, i]}| \cdot |T_{v, i + 1}|\right).
\]

The quantity $|T_{\child[v, i]}| \cdot |T_{v, i + 1}|$ 
can be interpreted as the number of pairs of vertices $(v_1, v_2)$ 
such that $v_1 \in T_{\child[v, i]}$ and 
$v_2 \in T_{v, i + 1}$. Note that
for all such pairs, the lowest common ancestor of $v_1$ and 
$v_2$ is $v$. Thus, if we sum this quantity over all $i\in [n_v]$, we get precisely the 
number of unordered pairs of distinct vertices whose lowest common ancestor is $v$. 
It is now immediate that, if we 
further sum this quantity over all $v\in V$, we get precisely $n \choose 2$, which is the 
number of unordered pairs of distinct nodes in a tree with $n$ vertices, completing the proof.
\end{proof}

By summing up the $O(n^2)$ terms from Lemma~\ref{lem:n_squared} over all $c\in C$, 
and observing that CCTP becomes trivial if $k\ge n$ (we can then afford to include the top
choice of each voter), we obtain the following bound.
\begin{theorem} 
Solving all subproblems of the form $\dyp_2[-, -, -, -]$ using Lemma~\ref{lem:recurrence} 
takes time $O(n^2m)$. Hence, CCTP can be solved in time $O(n^2m)$.
\end{theorem}

We are now ready to prove the $O(nmk)$ time bound. Just as in Lemma~\ref{lem:n_squared}, 
it suffices to bound the time required to solve all subproblems 
of the form $\dyp_2[-, -, -, c]$ for each $c\in C$.

\begin{lemma}\label{lem:n_k} 
For each candidate $c\in C$ solving all subproblems of the form $\dyp_2[-, -, -, c]$ 
takes time $O(nk)$.
\end{lemma}
\begin{proof} 
Let us revisit the expression for the time $S$ needed to solve all subproblems of this form:
\begin{equation}\label{equ_sum}
S = \sum_{v\in V, 1 \leq i \leq n_v}
(\min\{k, |T_{\child[v, i]}|\} \cdot \min\{k, |T_{v, i + 1}|\}).
\end{equation}
Note that the pairs $(v, i)$ in the summation index correspond to the edges of the tree.
This observation suggests
a new way of computing $S$ based on incrementally 
building the tree starting from $n$ singleton vertices.
Namely, we start with $S = 0$ and an empty graph $G$ consisting of $n$ disconnected singleton vertices, 
and repeat the next two steps until $G$ becomes isomorphic to $T$:
\begin{enumerate}
    \item\label{item_2} 
        Pick an edge $\{v, v'\}$ of $T$ that has not been chosen before 
	(where $v'$ is a child of $v$ in $T$) and 
        connect $v$ and $v'$ in $G$. This edge corresponds to a pair 
        $(v, i)$ such that $v'$ is the $i$-th child of $v$. 
        We call this operation a \emph{$(v, i)$-join}. A $(v, i)$-join 
        can only take place if all $(v, i')$-joins with $i'>i$ have already been performed 
        and the connected component of $v'$ in $G$ is isomorphic to~$T_{v'}$.
    \item 
	Increase $S$ by $\min\{k, |T_{v'}|\} \cdot \min\{k, |T_{v, i + 1}|\}$.
\end{enumerate}
We observe that at each step of this procedure the graph $G$ is a forest, and each component tree
of $G$ is of the form $T_{v, i}$ for some $v\in V$ and $1\le i\le n_v+1$.
Moreover, valid orders of joining the connected components of $G$
correspond to valid orders of solving all the subproblems of the form $\dyp_2[-, -, -, c]$, 
and the final value of $S$ (given in equation (\ref{equ_sum})) does not depend
on the order selected. In particular, for the purposes of our analysis 
it will be convenient to split the process into two phases: in the first phase, we will only join
two connected components if each of them has at most $k$ vertices, and in the second phase
we will perform the remaining joins.
Accordingly, let $S_1$ and $S_2$ denote the amounts added to $S$ in the first and the second
phase, respectively, so that $S=S_1 + S_2$. We will now argue that $S_1=O(nk)$ and $S_2=O(nk)$.

\begin{claim}\label{lem:S1}
$S_1=O(nk)$.
\end{claim}
\begin{proof}
    At the end of the first phase, the graph $G$ is a forest consisting of $p$ trees
    $T_{u_1, i_1}, T_{u_2, i_2}, \ldots, T_{u_p, i_p}$. 
    This state of $G$ corresponds to 
    having solved all subproblems of the form 
    $\dyp_2[v, i, -, c]$ on which $\dyp_2[u_1, i_1, -, c], \ldots, \dyp_2[u_p, i_p, -, c]$ depend 
    (possibly indirectly, and including the problems themselves), 
    and no others. This is the same as performing the complete algorithm restricted to each of 
    $T_{u_1, i_1}, T_{u_2, i_2}, \ldots, T_{u_p, i_p}$ individually. Thus, we can bound
    $S_1$ by applying Lemma~\ref{lem:n_squared} to each connected component 
    and summing up the results:
    \[
    S_1 \leq |T_{u_1, i_1}|^2 + |T_{u_2, i_2}|^2 + \ldots + |T_{u_p, i_p}|^2.
    \]
    Since each such connected component has been generated by joining two connected components 
    of size at most $k$, we can bound their individual sizes by $2k$. It follows that
    \[
    S_1 \leq 2k \cdot (|T_{u_1, i_1}| + |T_{u_2, i_2}| + \ldots + |T_{u_p, i_p}|) \le 2nk.
    \]
\end{proof}    

\begin{claim}\label{lem:S2}
$S_2=O(nk)$.
\end{claim}
\begin{proof}
    Given a sequence of $p$ integers $(a_1, \dots, a_p)$,
    let 
    \[
        \lambda(a_1, \dots, a_p) = \min\{k, a_1\} + \min\{k, a_2\} + \ldots + \min\{k, a_p\}.
    \]
    Suppose that at the start of the second phase the graph $G$ has $s$ components, 
    and let $(b_1, \dots, b_s)$ be the list of sizes of these components. 
    Note that $b_1+\dots+b_s=n$ and hence $\lambda(b_1, \dots, b_s)\le n$.
    Further, suppose that at some point during the second phase the list
    of sizes of the components of $G$ is given by $(f_1, \dots, f_q)$, and
    consider a join operation merging together two connected components of sizes $f_i$ and $f_j$. 
    At least one of these components has size greater than $k$;
    without loss of generality, assume that $f_i > k$. This operation removes $f_i$ and $f_j$ from 
    the list $(f_1, \dots, f_q)$. This changes the value of $\lambda$ by removing a 
    $\min\{k, f_i\} + \min\{k, f_j\} = k + \min\{k, f_j\}$ term and adding back a 
    $\min\{k, f_i + f_j\} = k$ term, thus decreasing 
    $\lambda$ by $\min\{k, f_j\}$. On the other hand, this operation increases $S_2$ by 
    $\min\{k, f_i\} \cdot \min\{k, f_j\} = k \cdot \min\{k, f_j\}$. Therefore, 
    whenever $\lambda$ decreases by $\Delta$, $S_2$ increases by 
    $k\Delta$. Since $\lambda$ can only ever decrease, starts off bounded from above by $n$ and never 
    becomes negative, $S_2$ is bounded from above by $nk$, completing the proof.    
\end{proof}
Now Lemma~\ref{lem:n_k} follows by combining Claims~\ref{lem:S1} and~\ref{lem:S2}. 
\end{proof}
As argued earlier in the paper, Lemma~\ref{lem:n_k} immediately
implies the desired bound on the performance of our algorithm.
\begin{theorem} 
{\sc CC-Winner-SCT} can be solved in time $O(nmk)$.
\end{theorem}
\section{Preferences Single-Crossing on a Grid}\label{sec:grid}

In this section we will discuss the challenges we face when trying to extend our results for 
{\sc CC-Winner} beyond trees. In particular, we present a useful lemma for preferences 
single-crossing on grids, which enables us to design a bicriterial approximation algorithm 
for {\sc CC-Winner-SCG}, as well as a polynomial-time algorithm under an 
additional (plausible) conjecture.

This section is structured as follows. First, we introduce the concept of laminar tilings for 
grid graphs. Then we provide a polynomial-time dynamic programming algorithm for {\sc 
CC-Winner-SCG} on grid graphs under the assumption that optimal solutions correspond to 
laminar tilings. Finally, we drop this assumption and argue that our dynamic programming 
algorithm can be transformed into a bicriterial approximation algorithm for our problem. 

Assume that $C = [m]$ and that the set of voters is $V = [n_1] \times [n_2]$, so that the 
preference profile $\calP=(\succ_v)_{v\in V}$ is single-crossing on the $n_1 \times n_2$ 
grid. For readability, we will write $\rho(i, j, c)$ to denote the dissatisfaction of voter 
$(i, j)$ with candidate $c$.

We begin by proving an analogue of Lemma~1 for grid graphs.
\begin{lemma}\label{lem:rectangle}
For every profile $\calP$ over $C$ that is single-crossing on the grid $V=[n_1]\times [n_2]$,
every $k\le |C|$, every canonical $k$-assignment $\opt$ for $\calP$ 
and every candidate $c$, it holds that the inverse image $w_{opt}^{-1}(c)$ defines a subrectangle of~$V$.
\end{lemma}
\begin{proof}
Let $R$ be the smallest rectangle that contains the set of voters $\opt^{-1}(c)$ (i.e.~their bounding box). Assume for the sake of contradiction that there is a voter $(i, j) \in R$ such that $\opt(i, j) = c' \neq c$. Then, there are two voters $v_0, v_1 \in R \cap \opt^{-1}(c)$ such that there is a shortest $v_0$--$v_1$ path $P$ passing through $(i, j)$. Since $\opt$ is a canonical assignment, it means that $c \succ_{v_0} c'$ and $c \succ_{v_1} c'$, but $c' \succ_{(i, j)} c$, contradicting the assumption that $\calP$ is single-crossing on $[n_1]\times [n_2]$. 
\end{proof}

Lemma~\ref{lem:rectangle} establishes that a canonical $k$-assignment $\opt$ can be viewed as 
a partition of the grid into a collection $\calF_{\mathit{opt}}$ of at most $k$ rectangles. 
All voters in a given rectangle $R \in \calF_{\mathit{opt}}$ share a common representative 
$c$, which minimizes their total dissatisfaction. Let us call any partition of $V$ into at 
most $k$ subrectangles a \emph{$k$-tiling} (i.e.~$\calF_{\mathit{opt}}$ is a $k$-tiling). 
Now, just as in Section~\ref{sec:tree},
where we reduced {\sc CC-Winner-SCT} to CCTP, 
from now on we will focus on the problem of finding a $k$-tiling $\calF_{\mathit{opt}}$ 
that minimizes the total dissatisfaction of the voters in the implicitly associated $k$-assignment.

Next, we define a class of tilings that are particularly convenient for our purposes.

\begin{definition}\label{def:sliceable}
Let $\calF$ be a $k$-tiling of $V$. We say that $\calF$ is \emph{laminar} 
if at least one of the following conditions holds:
\begin{itemize}
    \item $\calF$ consists of a single rectangle.
    \item $\calF$ can be partitioned into two sets of rectangles $\calF_{\leftarrow}$ and $\calF_\rightarrow$ such that $\calF_{\leftarrow}$ is a laminar tiling of $[n_1] \times [1 : j]$ and $\calF_\rightarrow$ is a laminar tiling of $[n_1] \times [j + 1 : n_2]$, for some $j \in [n_2 - 1]$.
    \item $\calF$ can be partitioned into two sets of rectangles $\calF_\uparrow$ and $\calF_\downarrow$ such that $\calF_\uparrow$ is a laminar tiling of $[1 : i] \times [n_2]$ and $\calF_\downarrow$ is a laminar tiling of $[i + 1 : n_1] \times [n_2]$, for some $i \in [n_1 - 1]$.
\end{itemize}
\end{definition}

Intuitively, $\calF$ is laminar if it can be recursively subdivided by vertical/horizontal lines until we get to singleton rectangles. We now state a conjecture, which we will then show to imply that {\sc CC-Winner-SCT} is polynomial-time solvable.

\begin{conjecture}\label{conj}
For every instance of {\sc CC-Winner-SCG} there exists an 
optimal $k$-tiling that is laminar.
\end{conjecture} 

There is empirical evidence that this conjecture is plausible: we have experimentally checked that it always holds for $n_1 \leq 4$, $n_2 \leq 5$ and $k \leq 5$, as well as for some other small instances. Additionally, for $k \leq 4$ there is a direct proof that our conjecture always holds. Furthermore, it can be shown to hold under the additional assumption that the preferences of every pair of voters that are adjacent in the grid differ in at most one pair of candidates. 

Assuming that the conjecture holds, we now propose a dynamic programming algorithm for {\sc CC-Winner-SCG}. We consider subproblems of the form $\dyp[i_0, i_1, j_0, j_1, \ell]$, representing the minimal possible dissatisfaction of voters in $[i_0 : i_1] \times [j_0 : j_1]$ for an $\ell$-tiling. These quantities are defined for $1 \leq i_0 \leq i_1 \leq n_1$, $1 \leq j_0 \leq j_1 \leq n_2$ and $1 \leq \ell \leq k$.

The following lemma presents the base case and recurrence relations that will be used to evaluate the dynamic program:

\begin{lemma} \label{grid_rec} 
Define the following three quantities:
\begin{align*}
\textit{CONST} = 
\min \{&\sum_{i = i_0}^{i_1}{\sum_{j = j_0}^{j_1}{\rho(i, j, c)}}: c\in [m]\};\\
\textit{VERT} = \min\{(&\dyp[i_0, i_1, j_0, j, \ell']\\
+ &\dyp[i_0, i_1, j + 1, j_1, \ell - \ell']):\\ 
&j_0 \leq j < j_1,  1 \leq \ell' < \ell\};\\
\textit{HOR} = 
\min\{(&\dyp[i_0, i, j_0, j_1, \ell'] \\
+ &\dyp[i + 1, i_1, j_0, j_1, \ell - \ell']):\\
&i_0 \leq i < i_1, 1 \leq \ell' < \ell\}.
\end{align*}
Then, it holds that 
$$
\dyp[i_0, i_1, j_0, j_1, \ell] = \min\{\textit{CONST}, \textit{VERT}, \textit{HOR}\}.
$$
\end{lemma}

\begin{proof} 
The proof is very similar to some of our earlier proofs, so we only sketch the main details. 
Observe that $\textit{CONST}$, $\textit{VERT}$ and $\textit{HOR}$ correspond to the three 
clauses of Definition~\ref{def:sliceable}. The base case $\textit{CONST}$ computes the 
optimal dissatisfaction attainable by a single rectangle. This is done by iterating through 
all possible candidates and finding the corresponding dissatisfaction in each case. 
$\textit{VERT}$ picks a vertical separation line and a number $\ell'$ and assumes that the 
tiling is split into $\ell'$ rectangles to the left of this line and $\ell - \ell'$ 
rectangles to its right. For $\textit{HOR}$ the reasoning is similar. This approach returns 
an optimal solution as long as Conjecture~\ref{conj} holds.
\end{proof}

Note that, just as in the case of trees, the recurrences in Lemma~\ref{grid_rec} do not forbid a given candidate from being assigned to more than one rectangle in a tiling. 

\begin{theorem}\label{thm:grid}
Assuming that Conjecture~\ref{conj} holds, {\sc SC-Winner-SCG} can be solved in time polynomial in $n_1$, $n_2$, $m$, $k$.
\end{theorem}
\begin{proof} We process the subproblems in increasing order of $i_1 - i_0 + j_1 - j_0 + \ell$, solving each subproblem by directly applying the recurrence relations in 
Lemma~\ref{grid_rec}. For the time complexity, consider a fixed subproblem and note that computing $\textit{VERT}$ and $\textit{HOR}$ together takes time $O(n_1k + n_2k)$ and that computing $\textit{CONST}$ takes time $O(n_1 n_2 m)$. Therefore, the total time complexity can be bounded as $O(n_1^2n_2^2k(n_1k + n_2k + n_1n_2m))$, which is polynomial in the input size. One can improve this bound by tighter analysis (in particular, by being more diligent about how $\textit{CONST}$ is computed), but we omit these details for now.
\end{proof}

The usefulness of Theorem~\ref{thm:grid} is limited, as we do not know at this point whether 
Conjecture~\ref{conj} is true. However, we will now show that we can transform our dynamic 
programming algorithm into a bicriterial approximation algorithm for our problem even without 
assuming the conjecture. Specifically, we can use our algorithm to find a $k^2$-assignment 
$w$ such that the voters' dissatisfaction under $w$ is at most their dissatisfaction under an 
optimal $k$-assignment. We begin by introducing some definitions and notation:

\begin{definition}\label{def:refine}
Let $\calF$ and $\calF'$ be two tilings of $V$. We say that $\calF'$ \emph{refines} $\calF$ if every rectangle in $\calF$ can be represented as a union of some rectangles in $\calF'$.
\end{definition}

The following proposition follows immediately from Definition~\ref{def:refine}.

\setcounter{proposition}{1}
\begin{proposition} Let $\calF$ and $\calF'$ be tilings of $V$ such that $\calF'$ refines $\calF$. Then the total dissatisfaction of $\calF'$ is at most that of $\calF$ (i.e.~refinement can not increase dissatisfaction).
\end{proposition}

\begin{definition} 
Given a tiling $\calF$ of $V$ and a vertical/horizontal line ${\mathcal L}$ between two columns/rows of the grid, let $\textit{split}_{ \mathcal{L}}(\calF)$ be the tiling obtained from $\calF$ by replacing each rectangle whose interior intersects $\mathcal{L}$ with the two rectangles that $\mathcal{L}$ splits it into.
\end{definition}

Observe that $\textit{split}_{\mathcal{L}}(\calF)$ refines $\calF$. Also, if $\calF$ is a $k$-tiling, then $\textit{split}_{\mathcal{L}}(\calF)$ is a $2k$-tiling.

Given a tiling $\calF$, we can think of its rectangles as geometric objects in $\mathbb{Z}^2$. In particular, a given rectangle $[i_0 : i_1] \times [j_0 : j_1] \in \calF$ has corners at coordinates $(x, y) \in$ $\{(i_0 - 1, j_0 - 1)$, $(i_0 - 1, j_1)$, $(i_1, j_0 - 1)$, $(i_1, j_1)\}$.\footnote{To keep consistent with the way grids are normally drawn, the $\textit{OX}$ axis should be regarded as going from north to south and the $\textit{OY}$ axis from east to west.} With this in mind, define $P(\calF)$ to be the set of all corners of all rectangles in $\calF$. Likewise, define $x(\calF) = \{x : (x, y) \in P(\calF)\}$ and $y(\calF) = \{y : (x, y) \in P(\calF)\}$.

\begin{lemma} \label{can_split_optimum}
Let $\calF$ be a $k$-tiling of $V$. 
Then there exists a laminar $k^2$-tiling $\calF'$ which refines $\calF$.
\end{lemma}
\begin{proof} 
Let $\ell_x = |x(\calF)|$, $\ell_y=|y(\calF)|$, and assume that $x(\calF) = \{x_1, \ldots, x_{\ell_x}\}$ and $y(\calF) = \{y_1, \ldots, y_{\ell_y}\}$, where $x_1 < \ldots < x_{\ell_x}$ and $y_1 < \ldots < y_{\ell_y}$.
Initialize $\calF'$ to $\calF$ and perform the following steps:
\begin{enumerate}
    \item For each vertical line $\mathcal{L}$ defined by an element of $y(\calF)$, replace $\calF'$ by $\textit{split}_{\mathcal{L}}(\calF')$.
    \item For each horizontal line $\mathcal{L}$ defined by an element of $x(\calF)$, replace $\calF'$ by $\textit{split}_{\mathcal{L}}(\calF')$.
\end{enumerate}

Clearly, $\calF'$ refines $\calF$, since each individual operation is a refinement. Furthermore, $\calF'$ consists of precisely $(\ell_x-1)(\ell_y-1)$ rectangles, each with a lower left corner at $(x_i, y_j)$ and an upper right corner at $(x_{i + 1}, y_{j + 1})$, for some $1 \leq i < \ell_x$ and  $1 \leq j < \ell_y$. Since $\calF$ is a $k$-tiling, it follows that $\ell_x \leq k + 1$ and $\ell_y \leq k + 1$, implying that $(\ell_x-1)(\ell_y-1) \leq k^2$, and so $\calF'$ is a $k^2$-tiling. By construction, $\calF'$ is laminar, hence completing the proof.
\end{proof}

An approximation guarantee now follows as a direct consequence:

\begin{corollary}\label{cor:bicriterial} 
By executing our dynamic programming algorithm with the target committee size set to $k^2$ instead of $k$, we are guaranteed to find a committee that is at least as good as the optimum one for $k$. 
\end{corollary}
Note that Corollary~\ref{cor:bicriterial} generalizes to $d$-grids with $d>2$. However, the 
required committee size becomes $k^d$, making this result less attractive for higher 
dimensions.

\section{Conclusions and Future Work}\label{sec:concl}
We have significantly improved the state of the art concerning the algorithmic complexity of 
the Chamberlin--Courant rule, both for preferences single-crossing on a line and for 
preferences single-crossing on a tree. For the former setting, the performance of our 
algorithms makes them suitable for a broad range of practical applications; for the latter 
setting, we identify an issue in prior work and present the first poly-time algorithm. It is 
instructive to contrast the algorithmic results for preferences single-crossing on trees and 
preferences single-peaked on trees: for the latter domain, positive results hold only if the 
underlying tree has a special structure, and the problem remains hard for general trees 
\cite{trees-full}, whereas our positive result holds for all trees. For the grids, while
we do not have an unconditional proof of correctness for our algorithm, we have collected
some evidence that our conjecture is true, and also found a way to leverage our approach
to design a bicriterial approximation algorithm.

In our paper, we focused on the utilitarian version of the Chamberlin--Courant rule, where 
the goal is to minimize the sum of voters' dissatisfactions; however, both of our $O(nmk)$ 
algorithms can be modified to compute winners under the {\em egalitarian} version of this 
rule, where the goal is to minimize the dissatisfaction of the most misrepresented voter, 
simply by replacing `+' with $\max$ in the respective dynamic programs. This is no longer the 
case for our reduction to the $k$-LPP problem; however, by using binary search to reduce the 
egalitarian problem to the utilitarian problem, we can nevertheless find solutions for the 
former in time $O(nm \log(n)\log(nm))$ using this approach.

It would be very interesting to extend our algorithmic results to general median graphs, 
or prove that such an extension is unlikely, by establishing an NP-hardness result. Resolving 
the conjecture about grid graphs would be a natural first step in this direction.

\bibliography{bib}

\appendix

\section{The Previous Algorithm for {\sc CC-Winner-SCT}}
\citeauthor{cps15}~[\citeyear{cps15}] present an algorithm for {\sc CC-Winner-SCT}
that proceeds by dynamic programming, 
building a solution for the entire tree from solutions for various subtrees. 
In this section we take a close look at their algorithm, 
showing how its runtime can be exponential in some cases.

We start by presenting the algorithm of \citet{cps15} (with a few typos corrected, and using 
our notation and terminology), and then show that its runtime can be exponential in the worst 
case, by exhibiting an explicit instance for which this occurs.

\subsection{A Recap of the Algorithm}
Fix a misrepresentation function $\rho$, and consider a profile $\calP=(\succ_v)_{v\in V}$, $|V|=n$, 
over $C=[m]$ that is single-crossing on a tree $T$. We will view $T$ as a rooted tree 
with $v_1$ as its root. A subtree $T' \subseteq T$ is a \emph{terminal subtree} 
if $T \setminus T'$ is also a subtree of $T$.

For any subtree of voters $T' \subseteq T$ with vertex set $V'$ and $a, b \in C$ such that 
$a\neq b$ we define $V'_{ab}=\{v \in V' : a \succ_v b\}$; let $T'_{ab}$ denote the subtree of 
$T'$ induced by $V'_{ab}$ (note that $T'_{ab}$ is a terminal subtree of $T'$ by the 
single-crossing property).

The following lemma proves key to the correctness of the algorithm.

\begin{customlemma}{4.5 (original number)}\label{arkadii_lemma} For some $k\in[m]$ 
let $\opt$ be a canonical $k$-assignment for $\calP$ and let $b \in C$ 
be the least favorite candidate of $v_1$. Then the vertices of $\opt^{-1}(b)$ 
form a terminal subtree of $T' \subseteq T$. Moreover, unless $T' = T$, 
we also have that $T' = T_{ba}$, for some $a\in C$ with $b \neq a$.
\end{customlemma}

See the original paper for the proof. Note that the lemma generalizes to subinstances 
where $T$ has been replaced by some $S \subseteq T$ that contains voter $v_1$ and/or 
some of the candidates in $C$ have been removed.
From now on, without loss of generality assume that $1\succ_{v_1}\ldots\succ_{v_1} m$.

For all $S \subseteq T$ such that $v_1 \in S$ and $j \in [m]$, $\ell \in [k]$ 
define $\dyp[S, j, \ell]$ to be the minimal dissatisfaction of voters in $S$ 
that can be achieved by selecting a size-$\ell$
committee out of candidates in $[j]$. For values of $S, j$ and $\ell$ other 
than the ones mentioned, we let $\dyp[S, j, \ell] = \infty$ for ease of explanation.

We omit the base cases; for the general case assume that $\ell \geq 2$, then following recurrence holds:
\begin{equation*}
\dyp[S, j, \ell] = \min\{\dyp[S, j - 1, \ell], X\},
\end{equation*}
where
\begin{equation*}
X = \min_{1 \leq a < j}{(\dyp[S_{aj}, j - 1, \ell - 1] + \sum_{v\in S_{ja}}{\rho(v, j)})}
\end{equation*}

Subproblems are solved in increasing order of $|S|$, breaking ties by increasing $j$. The overall minimum dissatisfaction is given by $\dyp[T, m, k]$ and a winning committee
can be computed using standard dynamic programming techniques.
The recurrence essentially stipulates that one will either not elect candidate $j$ in the sought size-$k$ committee (the $\dyp[S, j - 1, \ell]$ term), or they will, in which case candidate $j$ will represent a terminal subtree $S_{ja} \subseteq S$, as in Lemma~4.5 (the $X$ term).

The correctness of this approach is immediate from Lemma~4.5. On the other hand, 
it is already clear from this description that the number of values of $S$ may be exponential 
in $n$, leading to an exponential number of subproblems described above. However, this does 
not yet establish our claim: it might still be the case that only polynomially many 
subproblems would be considered in a top-down memoized implementation of the recurrence. We 
will now rule out this possibility, by presenting an explicit input instance where an 
exponential number of subproblems need to be considered.

\subsection{An Exponential Instance}
Our tree $T$ is a star graph, i.e., a graph with vertex set $V$, $|V| = n$, 
in which all vertices other than $v_1$ are leaves, and each leaf is connected to $v_1$. 
Let $v_1$ be the root vertex of this tree. Let $C = [n]$ and let $v_1$'s preferences 
be given by $1\succ_{v_1}\ldots\succ_{v_1} n$. The other voters' preferences 
are defined as follows: for each $i \in [2 : n]$, voter $v_i$ ranks candidate $i\in C$ first, 
followed by all other candidates, which are ranked in the same way as in $v_1$'s vote. 
It is immediate that the resulting profile is single-crossing on $T$. 
Figure~\ref{fig:counter} illustrates this construction for $n = 5$.

\begin{figure}
    \centering
    \includegraphics[scale=0.46]{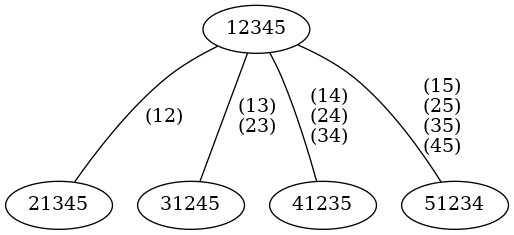}
    \caption{Voter $v_1$ at the top; voters $v_2, \ldots, v_n$ following below it, in order from left to right. For each voter $v$'s vertex, if its contents read $c_1c_2c_3c_4c_5$, then $c_1 \succ_{v} c_2 \succ_{v} c_3 \succ_{v} c_4 \succ_{v} c_5$. The labels on each edge denote preference \emph{cuts}. Namely, if an edge is labelled with $(cc')$, then that edge partitions $T$ into $T_{cc'}$ and $T_{c'c}$.} 
    \label{fig:counter}
\end{figure}

Following the execution of the algorithm, it is not difficult to see that for $k=n$ we will 
need to consider all subtrees of $T$ that contain $v_1$, proving our claim. Indeed, consider 
the first step of a recursive implementation of the recurrence, where $S = T$, $j = n$ and 
$\ell = k$. There are two possibilities: in one, $j$ is not picked to be part of the 
size-$\ell$ committee, in which case $S$ is left unchanged, $j$ is decremented and $\ell$ is 
left unchanged. In the other, $j$ is chosen to represent $S_{ja}$ for some $a \in [j - 1]$. 
By construction, observe that $S_{ja} = \{v_j\}$ regardless of the value of $a$, so $S$ 
changes to $S\setminus\{v_j\}$, $j$ is decremented and $\ell$ is decremented. Observe how $S$ 
is either left unchanged or $v_j$ is removed from it, with $j$ being decremented in both 
cases. This will happen recursively with $j - 1$, and so on, showing that, ultimately, $S$ 
will range over all subsets of $V$ that contain $v_1$ and have size at least $|V| - k + 1$. 
There are $\sum_{i = 0}^{k - 1}{n - 1\choose i}$ such subsets. For $k = n$, this quantity 
becomes $2^{n - 1}$, proving our claim.

Of course, one can argue that $k=n$ is a degenerate special case, as for this value of $k$ we can
simply include the top choice of every voter in the committee. However, the running time of 
the algorithm remains exponential for smaller values of $k$ as well: e.g., 
if $n$ is odd and $k=(n-1)/2$, we would need to consider all subtrees that 
can be obtained from $T$ by deleting at most half of the children of $v_1$, 
and there are $2^{n-2}$ such subtrees.

The original paper comes with an additional stipulation: ``pick the root $v_1$ to be a leaf 
in $T$'', which does not seem to be used implicitly or explicitly later on. In any case, to 
account for this, since our construction does not satisfy this requirement, one can add to 
$T$ an additional voter $v_1'$, duplicating $v_1$'s preferences, as well as an additional 
edge $v_1'$--$v_1$. Since duplicated preferences are also disallowed by the original paper, 
we can then add an additional candidate $m + 1$ to $C$. This candidate will be the most 
preferred candidate of $v_1'$, and the least preferred candidate of everyone else. With 
$v_1'$ now as the root of $T$, we obtain a tree that satisfies the assumptions of the 
original proof, and on which the running time is still exponential.

We leave open whether this algorithm can be revised so that only polynomially many 
subproblems need to be considered while preserving correctness.
\end{document}